\documentclass{article}
\usepackage[utf8]{inputenc}
\usepackage{algorithm}
\PassOptionsToPackage{noend}{algpseudocode}
\usepackage{algpseudocode}
\usepackage{amsmath,amssymb,amsfonts,amsthm}
\usepackage{tcolorbox}
\usepackage{fullpage}
\usepackage{hyperref}
\usepackage[capitalise]{cleveref}
\usepackage{tikz}
\usepackage{enumerate}

\algnewcommand\algorithmiccontinue{\textbf{continue}}
\algnewcommand\Continue{\algorithmiccontinue{} }

\newtheorem{theorem}{Theorem}

\newtheorem{lemma}[theorem]{Lemma}
\newtheorem{corollary}[theorem]{Corollary}

\theoremstyle{definition}

\newtheorem{observation}[theorem]{Observation}
\theoremstyle{remark}

\newcommand{\ceil}[1]{\left\lceil #1 \right\rceil}
\newcommand{\floor}[1]{\left\lfloor #1 \right\rfloor}

\newcommand{\nbd}[1]{\mathcal{N}[#1]}
\newcommand{\edges}{e}

\newcommand{\Prime}{\textsf{prime}}

\newcommand{\rb}[2]{\raisebox{#1 mm}[0mm][0mm]{#2}}
\newcommand{\istrut}[2][0]{\rule[- #1 mm]{0mm}{#1 mm}\rule{0mm}{#2 mm}}

\title{Maximum Matching on 
Regular Nonbipartite Graphs}

\author{Varsha Dani\\
Rochester Institute of Technology
\and 
Thomas P. Hayes\\
University at Buffalo
\and
Seth Pettie\thanks{Supported by NSF grant CCF-2446604.}\\
University of Michigan}

\date{}

\begin{document}
\maketitle

\begin{abstract}
\emph{Blocking flow}-type maximum matching algorithms 
are based on 
finding maximal sets of \emph{shortest augmenting paths}.
They run in 
$O(m\sqrt{n})$ time, on both bipartite~\cite{HopcroftK73,Dinic70,Karzanov73a,Karzanov73a}
and nonbipartite graphs~\cite{GabowT91,Gabow17,Vazirani24}, but this time bound can be improved if the 
input is constrained.
\smallskip 

In this paper we consider $d$-regular bipartite and nonbipartite graphs.  Previous algorithms show that a
perfect matching in $d$-regular bipartite graphs can be computed in near-linear time deterministically~\cite{ColeOS01}
or \emph{sub}linear time with high probability~\cite{GoelKK13}.  On $d$-regular non-bipartite graphs, a $(1-1/(d+1))$-approximation can be computed in sublinear time $O(n\log n)$ with high probability~\cite{DaniH25}, 
and hence a maximum matching 
can be computed in $O(n^2)$ time, w.h.p.,
which is slightly faster than the best deterministic algorithm for regular graphs~\cite{Yuster13}, 
running in $O(n^2\log n)$ time.

\smallskip 

We prove that \emph{any} blocking flow-type 
maximum matching algorithm based on finding 
shortest augmenting paths 
    runs in $O(n^2)$ time on $d$-regular graphs, both bipartite and nonbipartite.  
On nonbipartite graphs this is an asymptotic improvement over 
$O(n^2\log n)$~\cite{Yuster13} 
and an improvement over
$O(m\sqrt{n})$~\cite{GabowT91,Gabow17,Vazirani24}
when $d=\omega(\sqrt{n})$.
It also improves~\cite{DaniH25} 
by making its $O(n^2)$ bound deterministic.  However, the main take-away message is that no new algorithms are needed: the ``classic'' matching algorithms automatically outperform~\cite{Yuster13,DaniH25}.

\smallskip

We also consider extensions of our results to graphs that are only ``nearly regular,'' meaning that their degrees all lie within a specified range, $[d, \Delta].$
\end{abstract}

\section{Introduction}

The last few years has seen major progress in 
most of the canonical graph optimization problems, 
such as (min-cost) flows~\cite{ChenKLPGS25,BernsteinBLST25},
bipartite matching~\cite{ChuzhoyK24,ChuzhoyK24b},
vertex connectivity~\cite{LiNPSY25},
all-pairs min-cuts~\cite{AbboudKLPGSYY25},
edge-coloring~\cite{AssadiBBCSZ26},
and
shortest paths in non-negative real-weighted 
graphs~\cite{DuanMSY23,DuanMMSY25,DuanMSY26},
real-weighted graphs~\cite{Fineman24,HuangJQ26,KhannaS26,LiLRZ26,LiLZ26,HairLLZ26}, and integer-weighted graphs~\cite{BernsteinNW25,BringmannCF23,LiMR26,HaeuplerJS26}.  However, one problem that has for decades 
resisted all improvement 
is maximum matching in \emph{nonbipartite} graphs.  
The fastest algorithms for 
maximum cardinality matching are still
$O(m\sqrt{n})$~\cite{GabowT85,GabowT91,Gabow17,Vazirani24} and $O(n^{\omega})$~\cite{MuchaS04,Harvey09}, 
and the fastest
min-cost matching algorithms are still 
$\tilde{O}(m\sqrt{n}\log W)$~\cite{GabowT91,DuanPS18}
and $O(Wn^{\omega})$~\cite{CyganGS15}.

Maximum matching is known to be 
easier in many graph classes, such as 
planar, bounded-genus, and minor-free graphs~\cite{MuchaS06,YusterZ07}, 
random graphs~\cite{Motwani94,BastMST06},
various structured graph classes~\cite{HegerfeldK19},
and graphs with unique perfect matchings~\cite{GabowKT01},
to name a few examples.

\medskip 

In this paper we consider maximum cardinality matching 
in the class of \emph{$d$-regular graphs}.
This problem has an illustrious history 
going back to the 1970s; see \Cref{tab:history} for the history of this problem on bipartite and nonbipartite graphs.
On the bipartite side, the most notable results
are near-linear deterministic algorithms running in 
$O(m\log d)$~\cite{ColeOS01} 
and $O(m)+\tilde{O}(n)$ time~\cite{ColeH82},
and a \emph{sublinear} time \emph{randomized} 
algorithm running in $O(n\log n)$ time~\cite{GoelKK10}, independent of $d$.  
On the nonbipartite side, Yuster~\cite{Yuster13} 
gave a deterministic matching 
algorithm running in $O(n^2\log n)$ 
time for any $d$.  Yuster's algorithm also works for graphs where the maximum and minimum degree differ by at most $r$, finding a maximum matching in such graphs in time $O(r n^2 \log n)$.
Dani and Hayes~\cite{DaniH25} developed a randomized algorithm for nonbipartite $d$-regular graphs along the lines of Goel, Kapralov, and Khanna~\cite{GoelKK10} for bipartite graphs.
Their algorithm
finds a matching of size $\frac{n}{2}(1-\frac{1}{d+1})$ in $O(n\log n)$ time, with high probability.\footnote{Whereas regular bipartite graphs always have perfect matchings, every $d$-regular nonbipartite graph has a matching of size at least $\frac{n}{2}(1-\frac{1}{d+1})$, so both 
Goel, Kapralov, and Khanna~\cite{GoelKK10} and 
Dani and Hayes~\cite{DaniH25} work up to the existence threshold.}
This implies a randomized $O(n^2)$ time maximum 
cardinality matching algorithm for regular graphs.

\begin{table}[]
    \centering
    \begin{tabular}{|l|l|l|}
    \multicolumn{3}{l}{\textsc{Maximum Matching in $d$-Regular Biparitite Graphs}}\\
    \multicolumn{1}{l}{\textsc{Authors and Year}} 
&     \multicolumn{1}{l}{\textsc{Time Bound}} 
&     \multicolumn{1}{l}{\textsc{Notes}} \\\hline
    \istrut{4}Hopcroft \& Karp \hfill 1973 & \rb{-2}{$O(m\sqrt{n})$}  & \rb{-2}{Does not depend on $d$-regularity}\\
    \istrut[2]{3}Dinic \& Karzanov \hfill 1973 &                 &\\\hline
    \rb{-2.3}{Gabow \& Kariv} \hfill \rb{-2.3}{1978} & $O(m)$ \hfill only if $d=2^k$ & Maximum matching only\\
                            & $O(\min\{m\log^2 n,n^2\log n\})$ & Also edge-coloring\\\hline
    \rb{-2.3}{Cole \& Hopcroft} \hfill \rb{-2.3}{1982} & $O(m + n\log n(\log\log n)^2)$ & Maximum matching only\\
                                & $O(m\log n)$              & Also edge-coloring\\\hline
    \istrut[2]{4}Schrijver \hfill 1998 & $O(\Prime(d)m)$                 & Edge-coloring in $O((\Prime(d)+\log d)m)$\\\hline
    \istrut[2]{4}Cole, Ost, \& Schirra \hfill 2001 & $O(m\log d)$        & Also edge-coloring\\\hline
    \istrut[2]{4}Goel, Kapralov \& Khanna \hfill 2009 & $O(\min\{m,d^{-1}n^{2.5}\log n\})$ & \textsc{randomized}, always $\tilde{O}(n^{1.75})$\\\hline
    \istrut[2]{4}Goel, Kapralov \& Khanna \hfill 2009 & $O(\min\{m,d^{-1}n^2\log^3 n\})$ & \textsc{randomized}, always $\tilde{O}(n^{1.5})$.\\\hline
    \istrut[2]{4}Goel, Kapralov \& Khanna \hfill 2010 & $O(n\log n)$    & \textsc{randomized}, independent of $d$\\\hline
    \istrut[2]{4}\textbf{This work}\hfill 2026 & $O(n^2)$ & Standard matching algorithms\\\hline\hline
    \multicolumn{3}{l}{\ }\\
\multicolumn{3}{l}{\textsc{Maximum Matching in $d$-Regular Nonbipartite Graphs}}\\\hline
    \istrut{4}Micali \& Vazirani\hfill 1980 & & \\
    Gabow \& Tarjan\hfill 1991    & \rb{-2}{$O(m\sqrt{n})$}  & \rb{-2}{Does not depend on $d$-regularity}\\
    Gabow \hfill 2017             & &\\
    \istrut[2]{0}Vazirani\hfill 2024           & &\\\hline
    \istrut[2]{4}Yuster \hfill 2013  & $O(n^2\log n)$ & Independent of $d$\\\hline
    \istrut[2]{4}Dani \& Hayes \hfill 2025 & $O(n^2)$ & \textsc{randomized}, independent of $d$\\\hline
    \istrut[2]{4}\textbf{This work}\hfill 2026 & $O(n^2)$  & Standard matching algorithms\\\hline\hline
    \multicolumn{3}{l}{\ }\\
    \multicolumn{3}{l}{\textsc{Maximum Matching in Nearly Regular Graphs (with degrees in $[d,\Delta]$)}}\\\hline
    \istrut{4}Micali \& Vazirani\hfill 1980 & & \\
    Gabow \& Tarjan\hfill 1991    & \rb{-2}{$O(m\sqrt{n})$}  & \rb{-2}{Does not depend on near-regularity}\\
    Gabow \hfill 2017             & &\\
    \istrut[2]{0}Vazirani\hfill 2024           & &\\\hline
    \istrut[2]{4}Yuster \hfill 2013  & $O\big((\Delta+1-d)n^2\log n\big)$ & Depends on $(\Delta-d)$, not $d$\\\hline
    \istrut[2]{4}\textbf{This work}\hfill 2026 & $O((\Delta+1-d)n^2)$  & Standard matching algorithms\\\hline\hline
    \end{tabular}
    \caption{A history of maximum matching algorithms on $d$-regular graphs.  $\Prime(d) \in [2,d]$ is the largest prime factor of $d$. The result for Dani \& Hayes is corollary of~\cite[Theorem 1]{DaniH25}, but not claimed in \cite{DaniH25}.}
    \label{tab:history}
\end{table}

\subsection{New Results}

All prior algorithms for 
$d$-regular graphs were 
developed specifically 
for this graph class.
In this paper we prove that
\emph{any} blocking flow-type algorithm 
(such as Hopcroft-Karp~\cite{HopcroftK73}, 
Dinic-Karzanov~\cite{Karzanov73a,Karzanov73b},
Vazirani~\cite{Vazirani24},
or Gabow~\cite{Gabow17})
that finds maximal sets 
of shortest augmenting paths,
runs in $O(n^2)$ time on $d$-regular graphs, independent of $d$.\footnote{The Gabow-Tarjan~\cite{GabowT91} algorithm 
finds maximal sets of augmenting paths, but it is not claimed that they are \emph{shortest}.}

\medskip 

This result improves over Yuster~\cite{Yuster13} by a $\log n$ factor,
and improves Dani and Hayes~\cite{DaniH25} inasmuch as it is deterministic rather 
than randomized.  
Our result may be considered \emph{simpler} than~\cite{Yuster13,DaniH25} since we merely analyze standard, off-the-shelf matching algorithms 
rather than develop new ones.

\medskip 

At a technical level, we undertake a detailed study of the length and structure
of shortest augmenting paths in $d$-regular graphs with respect to a non-maximal matching $M$.  
We establish the following upper and 
lower bounds.

\begin{enumerate}[(i)]
    \item On bipartite graphs the shortest augmenting path has length $O(n/d)$, which is tight in the worst case.
    \item On nonbipartite graphs the shortest augmenting path can have length as large as $n-O(d^2)$ for $d=O(\sqrt{n})$ and $\Omega((n/d)^2)$ for $d=\Omega(\sqrt{n})$.
    \item On a nonbipartite graph with at least $n/(d+1)$ free vertices, the shortest augmenting path has length $O(n/d)$.
\end{enumerate}

Upper bound (i) implies blocking flow-type algorithms take $O(n^2)$ time on bipartite graphs.
Lower bound (ii) shows that no analogous result holds for general graphs, but upper bound (iii) is sufficient 
to prove that blocking flow-type algorithms
still take $O(n^2)$ time on nonbipartite graphs.

We then extend our results to the setting of ``nearly regular'' graphs. Specifically, we show that given a matching in a graph with vertex degrees in $[d, \Delta]$, 
if there are at least $\max\{ n/(d+1), n(\Delta -d)/(\Delta +d)\} $ free vertices, then, again, there is an augmenting path of length $O(n/d)$, implying that blocking flow algorithms run in time $O\big( (\Delta +1-d) n^2\big)$ time.  This beats Yuster's algorithm~\cite{Yuster13} whose running time on such graphs is $O\big( (\Delta +1-d) n^2 \log n\big)$, and it also beats the analysis of the running time of blocking flow algorithms when $\Delta-d < \frac{m}{n^{3/2}}$.

\subsection{Organization}

The analyses of blocking flow 
algorithms on $d$-regular bipartite and nonbipartite graphs appear
in \Cref{sect:bipartite}
and
\Cref{sect:nonbipartite-graphs},
respectively.
The lower bounds on shortest augmenting paths in $d$-regular nonbipartite graphs appear in \Cref{sect:lowerbound-proof}. The extension to nearly regular graphs is in Appendix~\ref{sect:nearreg}

\section{Bipartite Graphs}\label{sect:bipartite}

Let $G$ be a bipartite graph and $M$ be a matching in $G$. An alternating path structure 
is a graph rooted at a free vertex $s$,
whose nodes are 
partitioned into \emph{white} and \emph{black} 
levels.  $W_0=\{s\}$ is a white level,
and in general $W_i$ consists of all vertices at alternating path distance $2i$ from $s$, while $B_i$ consists of all vertices at alternating path distance $2i-1$ from $s$.
The construction of the alternating path structure terminates at the first black level $B_i$ that contains a free vertex, implying the existence of an augmenting path of length $2i-1$.
See \Cref{fig:bipartite-alternating} for an example.

\begin{observation}\label{obs:jumplevel}
    Note that unlike breadth-first search, the edges not part of shortest  alternating paths can jump across multiple levels. However, they can only do so in a limited fashion. In particular, vertices in $B_j$ have edges to $W_{j-1}$ and $W_j$ and may additionally have edges to any white layer $W_i$ with $i \ge j+1$.  Similarly, vertices in  $W_j$  have edges to $B_{j}$ and $B_{j+1}$, and may additionally have edges to any $B_i$ with $i \le j-1$.
\end{observation}

\begin{figure}
    \centering
    \includegraphics[scale=0.45]{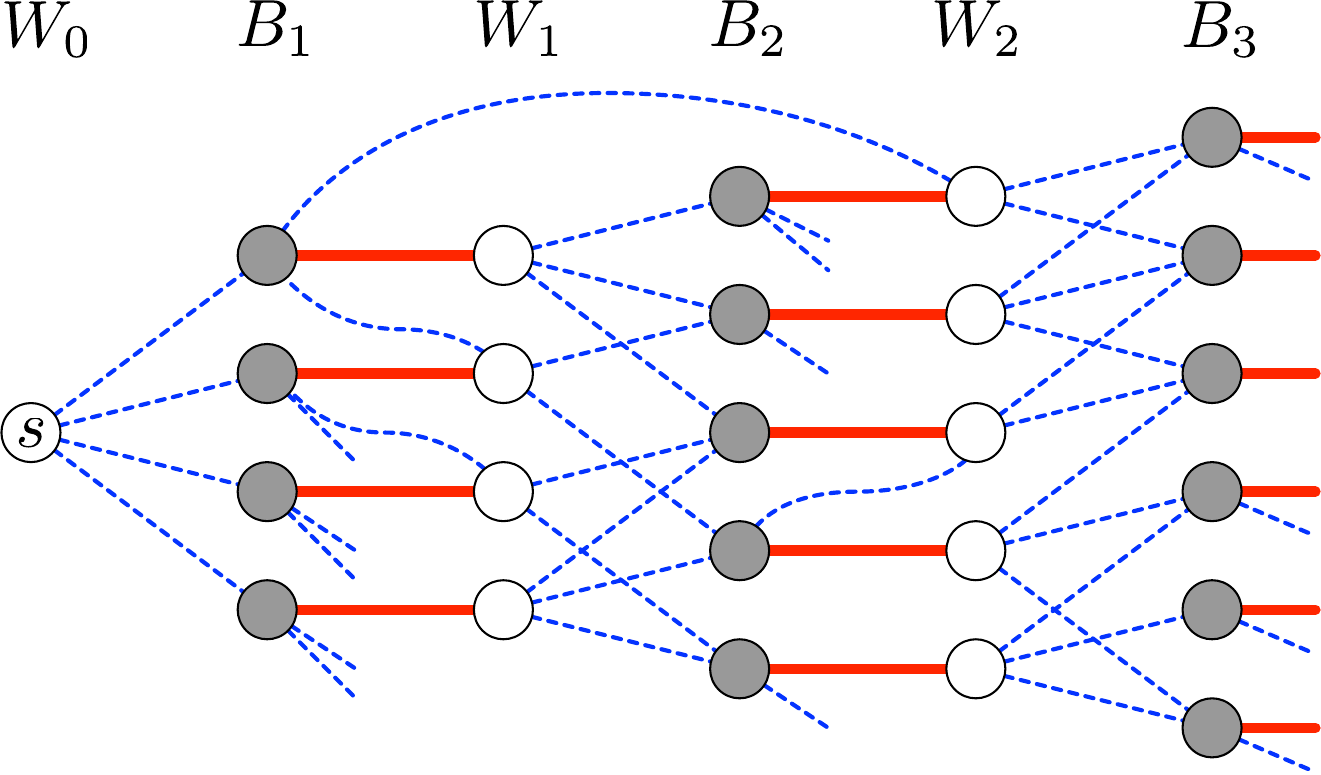}
    \caption{A search structure starting from free vertex $s$.  In this example $d=4$.  The sets $B_i,W_i$ are the vertices at alternating distance $2i-1$ and $2i$, respectively.}
    \label{fig:bipartite-alternating}
\end{figure}

\begin{lemma}\label{lem:excess}
    Let $G$ be a $d$-regular bipartite graph on $n$ vertices, and $M$ a non-perfect matching in $G$. Consider the alternating path structure rooted at a free white node $s$. Suppose there is no augmenting path of length at most $2j-1$. Let $G_j$ be the subgraph induced by the vertices in levels whose index is at most $j$, both black and white. Let $b_j$ be the number of edges going from black vertices in $G_j$ to white vertices outside $G_j$, and $w_j$ be the number of edges going from $W_j$ to  $B_{j+1}$.  Then $w_j = b_j +d$.
\end{lemma}
\begin{proof}
First note that by Observation~\ref{obs:jumplevel}, $w_j$ equals the number of edges going from white vertices in $G_j$ to black vertices outside $G_j$. Since there is no augmenting path of length at most $2j-1$, for each $1\le i \le j$, there are an equal number of vertices in $B_i$ and $W_i$ for each $i\ge 1$. Accounting for the root, $G_j$ has exactly one more white vertex than black. The claim now follows because $G$ is $d$-regular.
\end{proof}

\begin{theorem}\label{thm:augpath-bipartite}
    Let $G$ be a $d$-regular bipartite graph on $n$ vertices, and $M$ a non-perfect matching in $G$. Then there exists an augmenting path of length $O(n/d)$.
\end{theorem}

\begin{proof}
    Let $2\ell-1$ be the length of the shortest augmenting path. Build the alternating path structure starting from a free (white) vertex $s$. 
    For each $i < \ell$, let $n_i =|B_i| = |W_i|$. We will show that for any $i$ 
    with $1 < i < \ell-1$, 
    \[ 
    n_{i-1} + n_i + n_{i+1} > d. 
    \] 

\noindent Note that if $n_{i-1} + n_i \ge d$, then there is nothing to prove. 
    So we may assume $n_{i-1} + n_i < d$.

\medskip 

\begin{figure}
    \centering
    \includegraphics[scale=0.3]{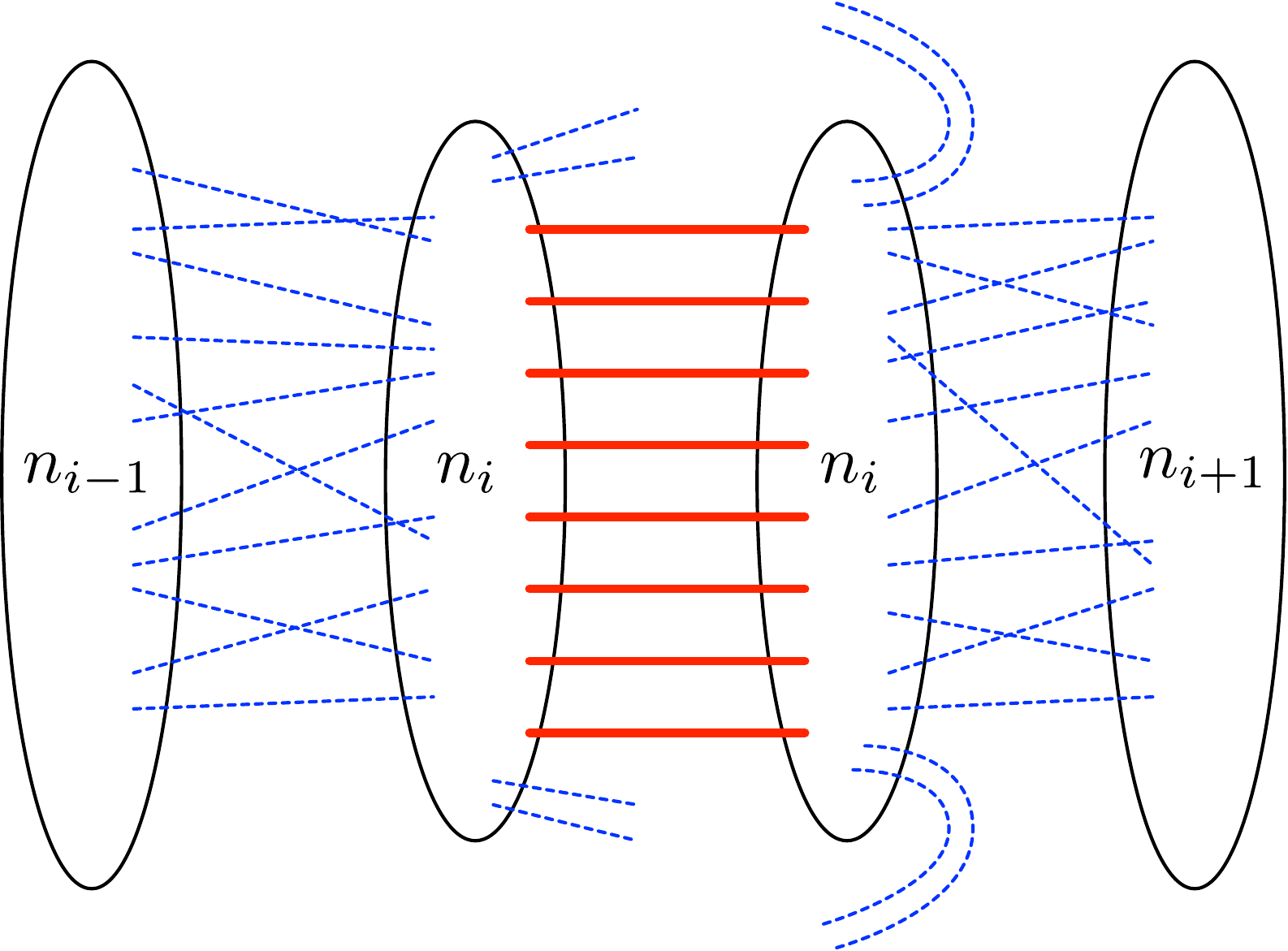}
    \caption{Four consecutive layers $W_{i-1},B_i,W_i,B_{i+1}$ with sizes $n_{i-1},n_i,n_i,n_{i+1}$, respectively.}
    \label{fig:bipartite-four-layers}
\end{figure}

Each vertex in $B_i$ has degree $d$, so there are $dn_i$ edges out of $B_i$. Since $ |W_i| = n_i $, at most $n_i^2$ of these edges go to $W_i$. Similarly, since $ |W_{i-1}| = n_{i-1}$, at most $n_i n_{i-1}$ of the edges out of $B_i$ go to $W_{i-1}$. Thus, at least $d n_i - n_i^2 - n_i n_{i-1}$ edges from $B_i$ go out of $G_i$. By 
Lemma~\ref{lem:excess}, at least $d n_i - n_i^2 - n_i n_{i-1} +d$ edges go from $W_i$ to $B_{i+1}$. By the pigeonhole principle, there is at least one vertex in $W_i$ with $d - n_i - n_{i-1} + d/n_i$ edges going to $B_{i+1}$. Since these edges emanate from a single vertex, they must go to distinct vertices in $B_{i+1}$ and therefore $n_{i+1} = |B_{i+1}| \ge d - n_i - n_{i-1} + d/n_i$. Rearranging terms, we have 
\[
n_{i-1} + n_i + n_{i+1} \ge d + d/n_i > d.
\]
as required.
It follows that whenever $M$ is not perfect,
there is an alternating path of length $O(n/d)$.
\end{proof}

Blocking flow-type matching algorithms~\cite{HopcroftK73,Karzanov73b,GabowT91,Gabow17,Vazirani24}
are organized in phases, where each phase finds a maximal set of minimum-length augmenting paths.  The key guarantee~\cite[Theorem 2]{HopcroftK73} is
that after $i$ phases, there are no augmenting paths of length less than 
$2i$.

\begin{corollary}\label{cor:bipartite-max-card-matching}
    Any blocking flow-type matching algorithm takes $O(n^2)$ time 
    when run on a $d$-regular bipartite graph.
\end{corollary}

\begin{proof}
\Cref{thm:augpath-bipartite} and \cite[Theorem 2]{HopcroftK73} implies that after phase $O(n/d)$, the matching must have maximum cardinality. 
Each phase takes linear time $O(dn)$ in the number of edges, so the overall time is $O(n^2)$.
\end{proof}

\section{Nonbipartite Graphs}\label{sect:nonbipartite-graphs}

The analogue of \Cref{thm:augpath-bipartite} 
for nonbipartite graphs 
is not true.  
The proof of \Cref{thm:nonbipartite-augpath-lower-bound} is simple, but somewhat long as we need
slightly different 
constructions depending on the 
magnitude of $d$ and its parity.
Refer to \Cref{sect:lowerbound-proof} 
for the proof.

\begin{theorem}\label{thm:nonbipartite-augpath-lower-bound}
    For infinitely many $n$ 
    and any $d\in [2,n)$
    there exists a $d$-regular, $n$-vertex graph $G$ and a non-maximum matching $M\subset E(G)$ for which the shortest augmenting path has length
\[
\left\{
\begin{array}{l@{\quad}l}
2\ceil{n/4}-1,   & \mbox{for $d=2$,}\\
n-1,             & \mbox{for $d=3$,}\\
n-O(d^2),        & \mbox{for $d=O(\sqrt{n})$,}\\
\Omega((n/d)^2), & \mbox{for $d=\Omega(\sqrt{n})$.}
\end{array}
\right.
\]
\end{theorem}

If the lower bound of \Cref{thm:nonbipartite-augpath-lower-bound} were tight when $d=\Omega(\sqrt{n})$, this would only imply a maximum matching algorithm runnning in $O(n^3/d)$ time, 
by the same reasoning found in \Cref{cor:bipartite-max-card-matching}. 
However, to achieve $O(n^2)$ time it suffices to show that a short $O(n/d)$-length augmenting path exists whenever there are sufficiently many free vertices, namely $n/(d+1)$.

\begin{theorem}\label{thm:augpath-nonbipartite}
Let $G$ be a $d$-regular graph on $n$ vertices and $M$ a matching in $G$ such that
\begin{equation}
    \label{eq:M-size}
    |M| \le \frac{n}{2}\left(1 - \frac{1}{d+1}\right).
\end{equation}
Then either $M$ is already a maximum matching or there exists an augmenting path of length $O(n/d)$.  
Moreover, when $M$ is maximum, 
\eqref{eq:M-size} holds with equality, $d$ is even, and $G$ is a disjoint union of $(d+1)$-cliques.
\end{theorem}

\begin{lemma}\label{lem:disjcliq}
    Let $G,M$ satisfy the conditions of \cref{thm:augpath-nonbipartite}.  If no vertex 
    is adjacent to two or more free vertices, then there is either an augmenting path of length at most 5, 
    or $d$ is even, $G$ is a union of disjoint $(d+1)$-cliques, and
    $M$ is a maximum matching in $G$.
\end{lemma}
\begin{proof}
For a vertex $v\in V$, let $\nbd{v}$ denote the set containing $v$ and all its neighbors. By $d$-regularity of $G$, for all $v$, $|\nbd{v}| = d+1$.
Let $W_0$ be the set of all free vertices in $G$ with respect to $M$. 
According to the preconditions of $G,M$, 
we have $|W_0| \ge \frac{n}{d+1}$. 

If there is no vertex that is adjacent to at least two free vertices, then by the maximality of $M$, the sets $\nbd{v}, v \in W_0$ are pairwise disjoint, and since there are at least $\frac{n}{d+1}$ of them, with $d+1$ elements each, they form a partition of $V$.  

Now consider $v \in W_0$. Since $M$ is maximal, all the neighbors of $v$ are matched. For a neighbor $w$ of $v$, if $M(w)$ is not a neighbor of $v$, then it must be a neighbor of some other free vertex $u$, so that $(v,w,M(w),u)$ 
is an augmenting path of length 3.

Now suppose for all free vertices $v$, the neighbors of $v$ are all matched to each other, which only holds when 
$d$ is even. For $u, v \in W_0$, if there is an edge from a neighbor of $u$ to a neighbor of $v$, then we get an augmenting path of length 5.

Finally if neither of the above cases hold, then for all $v \in W_0$, the subgraphs induced by $\nbd{v}$ are pairwise disjoint and by regularity, they are cliques. It follows that $G$ is a disjoint union of $(d+1)$-cliques, and $M$ is a maximum matching in $G$.   
\end{proof}

Let $G,M$ be as in the statement of \cref{thm:augpath-nonbipartite}.
For the remainder of the discussion, we will assume that $G$ has at least one vertex with two or more neighbors that are free with respect to $M$. 
Let $L \ge 1$ be a number such that $G$ contains no $M$-augmenting path of length less than $4(L+1)$; 
in particular $M$ is maximal.
Our goal is to compute an upper bound on $L$. 
To this end, we construct vertex sets $(W_i)_{i\geq 0}$ and $(B_i)_{i\geq 1}$
for $i$ up to $L$ as follows,
where $W_{<i}$ and $B_{<i}$ are short for 
$\bigcup_{j=0}^{i-1} W_j$ and $\bigcup_{j=1}^{i-1} B_j$, respectively.
\begin{align*}
W_0 &= \{\mbox{free vertices with respect to $M$}\},  & \mbox{Note: $|W_0| \ge n/(d+1)$ by assumption.} \\
S_i &= B_{<i} \cup W_{<i}, \\
B_i &= \left\{ v\in V\setminus S_{i} \mid v \mbox{ has \underline{\emph{at least two}} neighbors in } W_{< i} \right\},  \\
W_i &= \{ M(v) \mid v \in B_i \}.
\end{align*}

\begin{figure}
    \centering
    \includegraphics[scale=0.4]{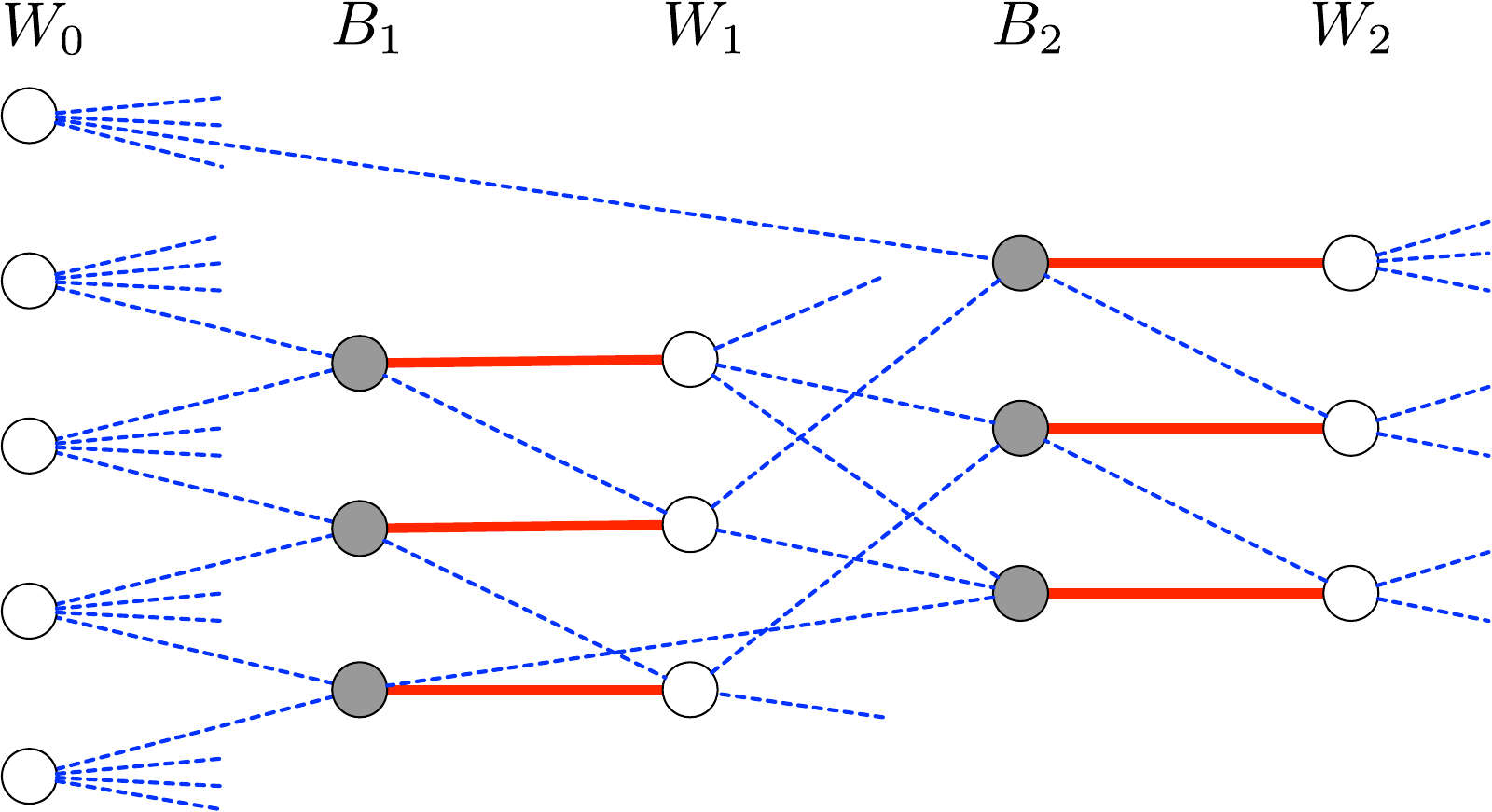}
    \caption{An example of the layered graph construction 
    for \emph{nonbipartite} graphs, with $d=4$.  
    Layer $W_0$ consists of the free vertices; 
    $B_1$ is the set of nodes with \emph{at least two neighbors} in $W_0$; $W_1$ is the set of matched neighbors of $B_1$; $B_2$ is the set of nodes with at least two neighbors in $W_0\cup W_1$, etc.  
    Note that one node of $B_2$ is there by virtue of having one neighbor in $W_0$ and one in $W_1$, while another node in $B_2$ has three neighbors in $W_0\cup B_1\cup W_1$, but its neighbor in $B_1$ is not relevant to its membership in $B_2$.}
    \label{fig:nonbipartite-example}
\end{figure}

See \Cref{fig:nonbipartite-example} for an example of this construction.

\begin{lemma}\label{lem:nonbip-properties}
    The construction satisfies the following properties.
    \begin{enumerate}
        \item $W_0, B_1, W_1, B_2, W_2, \dots, B_L, W_L$ are pairwise disjoint. \label{props-disjoint}

        \item For any $u_0, u_1 \in B_{< i}$ where $u_0 \ne u_1$, there exist vertex-disjoint
        alternating paths $P_0, P_1$ from $u_0, u_1$ to vertices in $W_0$, such that, furthermore, 
        the lengths of $P_0, P_1$ are both odd integers at most $2i-3$.
        Similarly, for any $u_0, u_1 \in W_{<i}$, where $u_0 \ne u_1$, there exist vertex-disjoint
        alternating paths $P_0, P_1$ from $u_0, u_1$ to vertices in $W_0$, such that, furthermore, 
        the lengths of $P_0, P_1$ are both even integers less at most $2i-2$.
        \label{props-item1}

        \item 
        $W = W_{\leq L}$ is an independent set.\label{props-item2}

        \item For $i\in[1,L]$, the subgraph induced by $B_i$ contains no matched edges.\label{props-item3}

        \item For $i\in[1,L]$, $|W_i|=|B_i|>0$.\label{props-item4} 
    \end{enumerate}
\end{lemma}

\begin{proof}
    We proceed by induction on $i$.  
    In the base case $i=1$:
    \begin{itemize}
        \item Property~\ref{props-disjoint} holds for $W_0$ and $B_1$ by the maximality of $M$.
        \item Property~\ref{props-item1} is immediate from the definition of $B_1$.
        \item For Property~\ref{props-item2}, note that $W_0$ is an independent set since if not, there would be an augmenting path of length 1.
        \item For Property~\ref{props-item3}, if $B_1$ contained a matched edge, 
        Property~\ref{props-item1} would imply
        an augmenting path of length 
        $3 < 4(L+1)$, contradicting the definition of $L$.
        \item Finally, following Lemma~\ref{lem:disjcliq}, since we have assumed that $G$ has at least one vertex with two or more neighbors that are free with respect to $M$, we have $|B_1| > 0$. Moreover, since $B_1$ does not contain a matched edge (Property~\ref{props-item3}), 
        $W_1$ is disjoint from $B_1$ and $|W_1|= |B_1|$.
    \end{itemize}

    Given that the claims hold for 
    $S_i$, layers indexed up to $i-1$, 
    we now prove that they also hold for $S_{i+1}$.
    For the black part of Property~\ref{props-disjoint}, note that $B_i$ is disjoint from $W_{<i}$ and $B_{<i}$ 
    by construction.
    Suppose $u_0\in B_i$ and $u_1\in B_j$, $j\leq i$.
    Because each black vertex has \emph{two} choices for a predecessor in previous white layers, there is a $v_1\in W_{j-1}$ adjacent to $u_1$, and an even-length alternating path from $u_0$ to some $v_0\neq v_1$, $v_0 \in W_{0} \cup \cdots \cup W_{j-1}$, 
    that avoids $u_1$.
    By the inductive hypothesis there are vertex-disjoint 
    alternating paths from $v_0$ and $v_1$ back to $W_0$,
    which implies vertex-disjoint paths from $u_0,u_1$ back to $W_0$.
    This confirms the black part of Property~\ref{props-item1}.
    
    Property \ref{props-item3} holds, since if $u_0,u_1\in B_i$ and 
    $\{u_0,u_1\}\in M$ then by Property~\ref{props-item1} there would be an augmenting path of length at most $4i-1 < 4L$, 
    contradicting the definition of $L$.
    This implies that $W_i\cap B_i=\emptyset$ and $|B_i|=|W_i|$, confirming part of Property~\ref{props-item4}. To see that $W_i$ is disjoint from $S_{<i}$, note that we have shown, for all $j \le i$, each vertex in $W_j$ is paired by $M$ with a vertex in $B_j$.  So disjointness of $W_i$ from $S_{<i}$ follows from disjointness of $B_i$ from $S_{<i}$, which was established earlier in this proof.
    Property~\ref{props-item1} for white vertices 
    $u_0,u_1 \in W_{<i}$ then follows, as we can apply
    Property~\ref{props-item1} to the (distinct) black vertices 
    $v_0 = M(u_0)$ and $v_1=M(u_1)$
    in $B_{<i}$.
    
    Property~\ref{props-item2}, that $W_{<i}$ is an independent set, follows the same reasoning as Property~\ref{props-item3}.
    Suppose $\{u_0,u_1\}\in E$ exists with $u_0,u_1\in W\cap S_i$.
    Property~\ref{props-item1} implies there are vertex-disjoint alternating paths from $u_0,u_1$ back to free vertices in $W_0$, 
    and together with $\{u_0,u_1\}$ this forms an augmenting path of length at most $4i+1 \leq 4L+1$, 
    contradicting the definition of $L$.

    We have established Properties \ref{props-disjoint}--\ref{props-item4}
    hold for $W_{<i}$ and $B_{<i}$, $W_i$, $B_i$, except for the non-emptiness criterion $|B_i|,|W_i|>0$ from Property~\ref{props-item4}.
    By the inductive hypothesis of Property~\ref{props-item4},
    $|W_{<i}| - |B_{<i}| = |W_0| \ge n/(d+1)$,
    and by Property~\ref{props-item2} all edges incident to $W_{<i}$ have their other endpoint in $B_{<i}$ or outside $S_{i}$.  By regularity, there must be at least 
    $d|W_0| \geq dn/(d+1)$ edges from $W_{<i}$ to 
    $V\setminus S_{i}$, but since we already established in the base case that $B_1 \ne \emptyset$,
    $|V\setminus S_{i}| < n - |W_0| 
    \leq nd/(d+1)$,
    so by the pigeonhole principle, there exists some $v\in (V\setminus S_{i-1})$ adjacent to two vertices in $W\cap S_{i-1}$, hence $|B_i|>0$, which establishes the final claim of 
    Property~\ref{props-item4}.
    \end{proof}

In light of \cref{lem:nonbip-properties}(4) we define
$n_0 = |W_0|$ and $n_i = |B_i| = |W_i|$, 
for $1 \le i \leq L$.

\begin{lemma} \label{lem:three-layer-lower-bound}
Let $i < L$.  Then $n_{i-1} +n_i +n_{i+1} \ge d+1$.
\end{lemma}
\begin{proof}
Fix $i$ and let  
$B'=B_1\cup\cdots\cup B_{i-1}$ 
and 
$W'=W_0\cup\cdots\cup W_{i-1}$.
By \cref{lem:nonbip-properties}(\ref{props-item4}), $|W'| = |B'| +n_0$,
so the sums of the degrees of vertices in $W'$ and $B'$ differ by $dn_0$. 
Let $R = V\setminus (B'\cup W')$,
and let the number of edges crossing the 
cut
$(B'\cup W', R)$ be $x+y$, 
where $x$ is the number incident to $B'$-vertices 
and $y$ the number incident to $W'$-vertices. 
\cref{lem:nonbip-properties}(\ref{props-item2}) says that $W$ is an independent set, which implies that
\begin{equation}\label{eq:y}
    y \ge x + dn_0.
\end{equation}
The sets $B_i, W_i$ and $B_{i+1}$ are disjoint and all contained in $R$. 
Let $R' = R\setminus (B_i \cup W_i\cup B_{i+1})$ be the rest of $R$, 
outside of these sets.  See \Cref{fig:nonbipartite-proof}.

\begin{figure}
    \centering
    \includegraphics[scale=0.35]{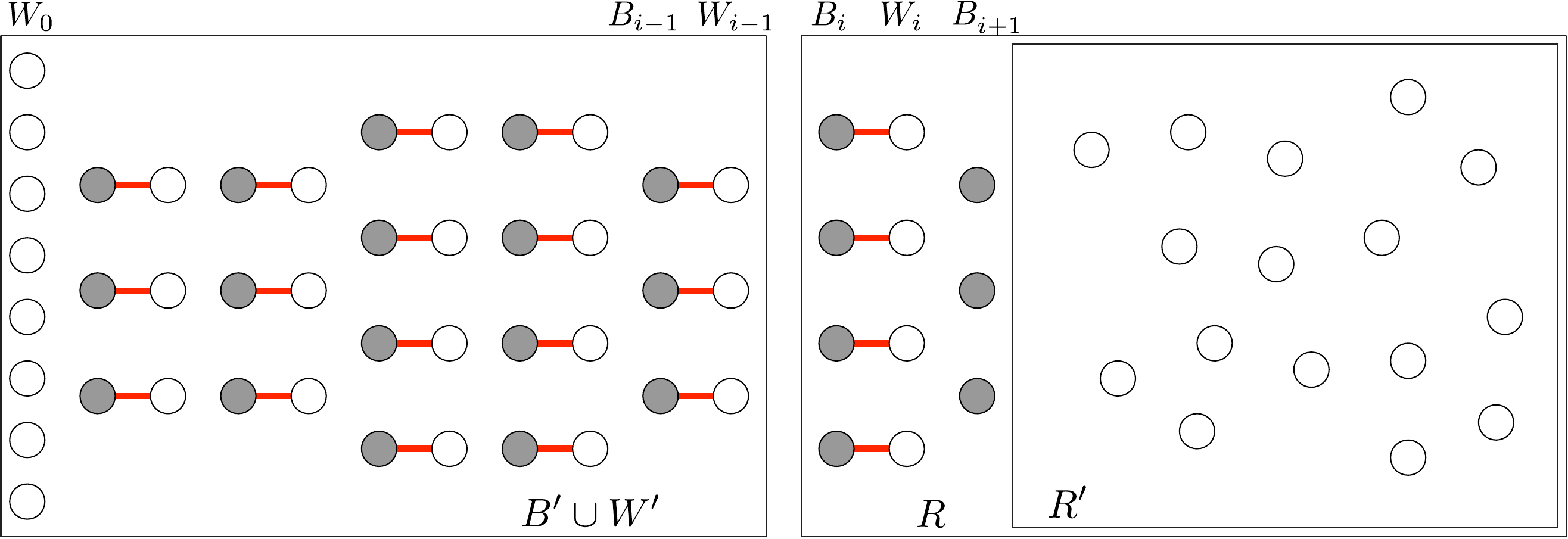}
    \caption{By definition 
    $B'=B_1\cup\cdots\cup B_{i-1}$
    and 
    $W'=W_0\cup\cdots \cup W_{i-1}$
    are the vertices in layers indexed up to $i-1$,
    $R=V\backslash (B'\cup W')$ is the rest,
    and $R' = R\backslash (B_i\cup W_i \cup B_{i+1})$ is the rest of the rest.
    The number of edges from $B'$ to $R$ is $x$, while the number of edges from $W'$ to $R$ is $y$.}
    \label{fig:nonbipartite-proof}
\end{figure}

By regularity and \cref{lem:nonbip-properties}(\ref{props-item2}), the number of edges with exactly one endpoint in $W_i$ is $dn_i$.
Moreover, by \cref{lem:nonbip-properties}(\ref{props-item2}) again, the other endpoints of these edges must be in $B', B_i, B_{i+1},$ or $R'$.
At most $n_i^2$ and $n_i n_{i+1}$ of these edges go to $B_{i}$ and $B_{i+1}$, respectively. 
The number of edges from $W_i$ to $B'$ is at most $x$. 
By definition each vertex in $B_{i+1}$ is adjacent to at least \emph{two} vertices in $W'\cup W_i$, 
which implies that any vertex in $R'$ can be adjacent to at most one neighbor in $W_i$.  
Letting $r'$ be the number of edges 
from $R'$ to $W_i$, we have
\begin{equation}\label{eq:dni}
    dn_i \le n_i^2 + n_i n_{i+1} + x + r'.
\end{equation}

We will now account for the $y$ edges joining 
$W'$ and $R$.  By \cref{lem:nonbip-properties}(\ref{props-item2}) no edges join $W'$ and $W_i$.
We can classify these edges into four groups based on the sets containing their endpoints, as follows.  
\begin{description}
    \item[Edges joining $W_{i-1}$ and $B_i$.] There are at most $n_{i-1}n_i$ of these.
    \item[Edges joining $W'\setminus W_{i-1}$ and $B_i$.] These edges must be attached to distinct vertices in $B_i$, for any vertex with two neighbors in $W'\setminus W_{i-1}$ would have been put in $B_j$ for some $j<i$. 
    There are at most $n_i$ of these.
    \item[Edges joining $W'$ and $B_{i+1}$.]  For the same reason these edges must be attached to distinct vertices in $B_{i+1}$, so there can be at most $n_{i+1}$ of these.
    \item[Edges joining $W'$ and $R'$.]  Once again
    no vertex in $R'$ can be attached to two edges
    back to $W'\cup W_i$, for otherwise it would be put in $B_{i+1}$.  There were $r'$ edges joining $R'$ and $W_i$, so there can be at most $|R'|-r'$ edges joining $R'$ and $W'$.
\end{description}
We can therefore bound $y$ as follows.
\begin{equation}\label{eq:yub}
    y \le n_{i-1}n_i + n_i + n_{i+1} + |R'|-r'.
    \end{equation}
Combining Equations~\eqref{eq:y},~\eqref{eq:dni} and~\eqref{eq:yub}, we get
\begin{align}
    dn_i &\le n_i^2 + n_i n_{i+1} + x + r' & \text{\cref{eq:dni}}\nonumber\\
    &\le n_i^2 + n_i n_{i+1} + y - dn_0 +r' & \text{\cref{eq:y}}\nonumber\\
    &\le n_i^2 + n_i n_{i+1} + n_{i-1}n_i + n_i + n_{i+1} + |R'|-r'- dn_0 +r' & \text{\cref{eq:yub}}\nonumber\\
    &\le n_i(n_{i-1}+n_i +n_{i+1}) + (n_i +n_{i+1} + |R'|) -dn_0.\label{eq:comb123}
\end{align}

Now recall that $2n_i +n_{i+1} + |R'| = |R| \le n-n_0$, and also that $n_0 \ge n/(d+1)$. 
Plugging these inequalities back into Equation~\eqref{eq:comb123}, we get
\begin{align*}
    dn_i &\le n_i(n_{i-1}+n_i +n_{i+1}) + (n_i +n_{i+1} + |R'|) -dn_0 \\
    &\le n_i(n_{i-1}+n_i +n_{i+1}) + n-n_i -n_0 -dn_0\\
    &= n_i(n_{i-1}+n_i +n_{i+1} -1) + n-(d+1)n_0 \\
    &\le n_i(n_{i-1}+n_i +n_{i+1} -1).
\end{align*}
Dividing by $n_i$, it follows that 
\[
n_{i-1}+n_i +n_{i+1} \ge d+1,
\]
which concludes the proof.
\end{proof}

\begin{proof}[Proof of Theorem~\ref{thm:augpath-nonbipartite}]
    Since the layers in our construction are disjoint, we have
    \[
    n \ge n_0 + \sum_{i=1}^{L-1} 2 n_i
    \]
    By Lemma~\ref{lem:three-layer-lower-bound}, it follows that
    \[
    n \ge 2(d+1) \left\lfloor{\frac{L-1}{3}}\right\rfloor,
    \]
    which gives the desired upper bound on the length of an augmenting path.
\end{proof}

\begin{corollary}\label{cor:nonbipartite-max-card-matching}
    Any blocking flow-type matching algorithm takes $O(n^2)$ time 
    when run on a $d$-regular nonbipartite graph.
\end{corollary}

\begin{proof}
After $c_0(n/d)$ phases of a blocking flow-type matching algorithm, the shortest augmenting path has length
at least $2c_0(n/d)$.
For $c_0$ sufficiently large,
\Cref{thm:augpath-nonbipartite} 
implies that there must be fewer than $\frac{n}{d+1}$ free vertices, and therefore
fewer than $\frac{n}{2(d+1)}$ 
additional phases.  With each phase taking linear time\footnote{See Gabow and Tarjan~\cite{GabowT85} 
for a \emph{linear time} 
implementation of the union-find routine used in the augmenting path algorithms of~\cite{GabowT91,Gabow17,Vazirani24}.} 
in the number of edges, $O(n/d)$ phases take $O(n^2)$ time.
\end{proof}

\bigskip 

\paragraph{AI Acknowledgment}
Generative AI tools were used for 
mathematical copyediting,
but not for generating any text of proofs.

\bibliographystyle{alpha}
\bibliography{refs}

\appendix

\section{Proof of \Cref{thm:nonbipartite-augpath-lower-bound}}\label{sect:lowerbound-proof}

There are six constructions:
two for $d\in\{2,3\}$,
and two each for small $d=O(\sqrt{n})$
and large $d=\Omega(\sqrt{n})$,
depending on the parity of $d$.
We will present them in order of complexity:
\begin{center}
    $d=2$
    $\;\;<\;\;$
    $d=3$
    $\;\;<\;\;$
    odd $d=O(\sqrt{n})$ 
    $\;\;<\;\;$ 
    even $d=O(\sqrt{n})$ 
    $\;\;<\;\;$ 
    odd $d=\Omega(\sqrt{n})$ 
    $\;\;<\;\;$ 
    even $d=\Omega(\sqrt{n})$.
\end{center}

\paragraph{Construction for $d=2$.} $G$ is an $n$-cycle and $M$ leaves two free vertices,
which are as close to antipodal as possible.  
The shortest augmenting path has length $2\ceil{n/4}-1$.

\paragraph{Construction for $d=3$.} $G,M$ are as depicted in \Cref{fig:d-three}.
It is straightforward to verify that there are $2^{\Omega(n)}$
augmenting paths, all of which use every vertex in the graph.

\begin{figure}[h]
    \centering
    \includegraphics[scale=0.34]{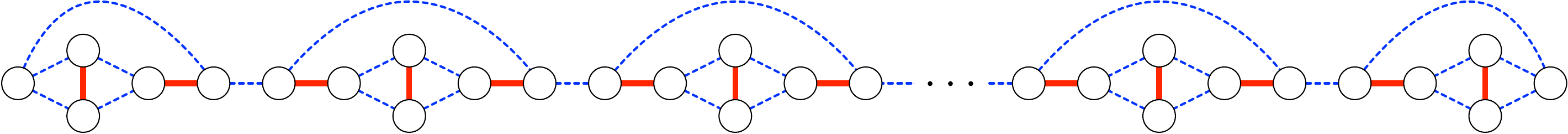}
    \caption{The $d=3$ graph.}
    \label{fig:d-three}
\end{figure}

\paragraph{Construction for small, odd $d$.} 
An odd-$d$ \emph{balloon gadget} consists of
vertices $u,u',v_1,\ldots,v_{d+1}$ and the following edges.  See \Cref{fig:balloon}.
\begin{itemize}
    \item The matching $M$ includes $\{u,u'\}$
    and $\{v_1,v_2\},\{v_3,v_4\},\ldots,\{v_d,v_{d+1}\}$.
    \item There are unmatched edges from $u'$ 
    to $v_1,\ldots, v_{d-1}$, and $(d(d+1) - (d-1))/2$
    edges between $v_1,\ldots,v_{d+1}$ to give them all degree $d$.
    \item The ``free degree'' of $u$ is $d-1$, which we can regard as $d-1$ half-edges, whose other endpoint will be assigned to some vertex outside the balloon gadget.
\end{itemize}

\begin{figure}
    \centering
    \includegraphics[scale=0.35]{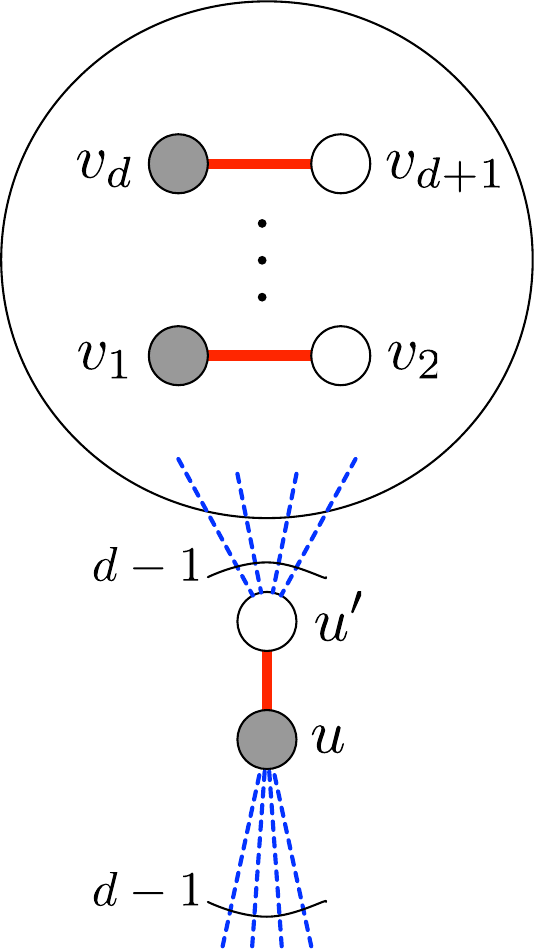}
    \caption{An odd-$d$ balloon gadget.}
    \label{fig:balloon}
\end{figure}

The construction of $G,M$ is parameterized by $b$, 
the number of \emph{blocks}, $k$, controlling the size of a block, and $t$, the number of balloon gadgets.  

\medskip 

Let $P$ be an augmenting path consisting of 
$2 + 2bk$ vertices.  
Excluding the two free endpoints $x,x'$,
the remainder of $P$ is partitioned into $b$ blocks $B_1,\ldots,B_b$, 
each of $2k$ vertices.  
Call a vertex in $P$ \emph{even} or \emph{odd} (white or gray in \Cref{fig:odd-small-d}) if it is at an even or odd 
distance from $x$ in $P$.

\begin{figure}
    \centering
    \includegraphics[scale=0.3]{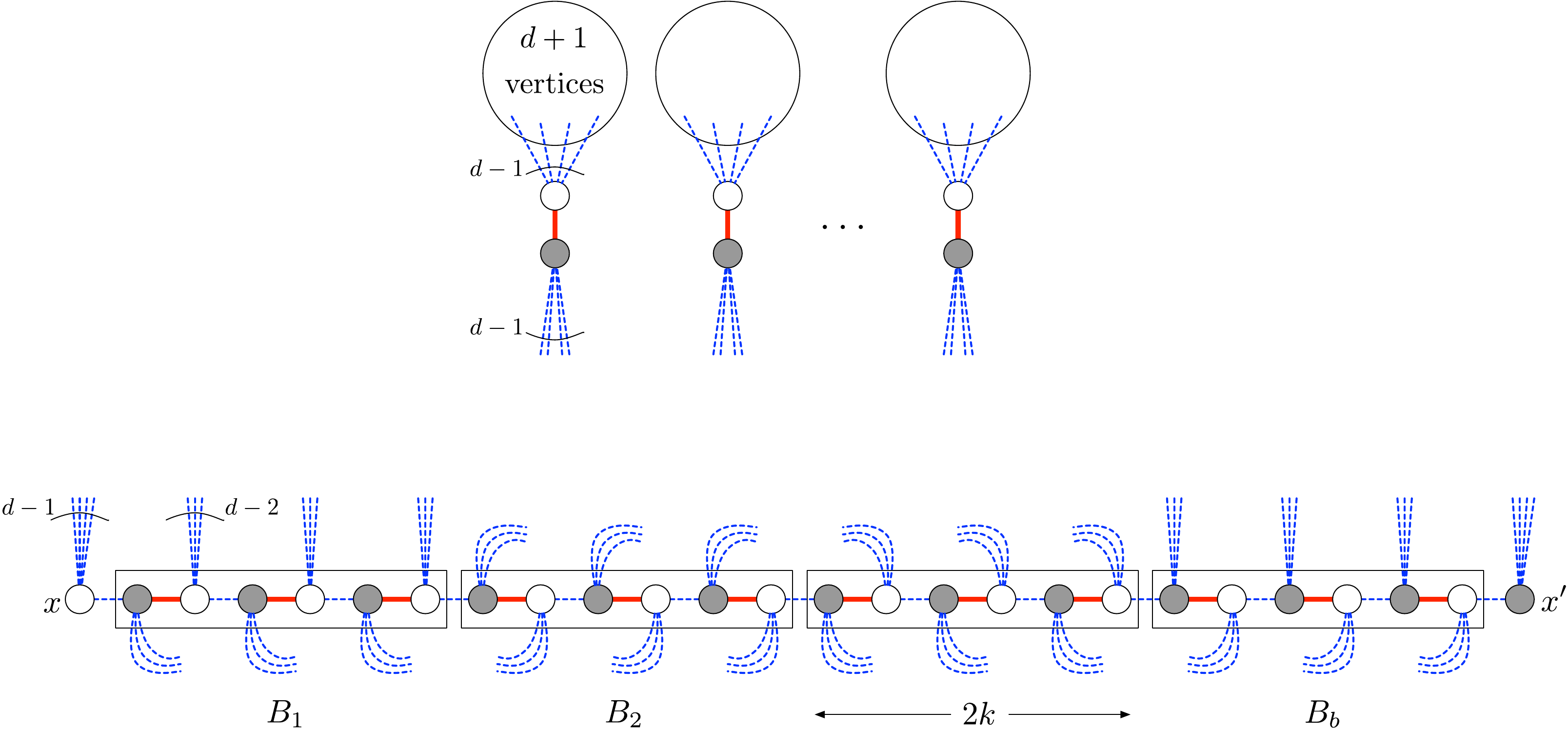}
    \caption{The graph for odd $d=O(\sqrt{n})$.}
    \label{fig:odd-small-d}
\end{figure}

For each $i$ from 1 to $b-1$, add edges between the odd vertices of $B_i$ to the even vertices of $B_{i+1}$ to make them all have degree $d$.  Since each such vertex has degree 2 in $P$, we need to add $d-2$ edges per vertex, which is possible whenever 
$k\geq d-2$.

At this point the free degree of 
$x$ and $x'$ are each $d-1$, 
the free degree of the even vertices in $B_1$ and the odd vertices in $B_b$ are each $d-2$,
and the free degree of each of the $t$ balloon gadgets is $d-1$.  We want to pair up the half-edges of the balloon gadgets to the half-edges of the $V(P)$-vertices, 
and therefore want:
\[
2(d-1 + k(d-2)) = t(d-1),
\]
which is possible by setting 
$k=d-1$ and $t=2(d-1)$.

One may confirm that there is only one augmenting path from $x$ to $x'$, namely $P$.
An alternating path from $x$ that enters a balloon gadget cannot return to $P$, 
and any alternating path from $x$ to $x'$ using $V(P)$ vertices cannot use any of the
shortcut edges between $B_i$ and $B_{i+1}$.

The total number of vertices is
\[
|V(P)| + t(d+3) = |V(P)| + 2(d-1)(d+3),
\]
implying the shortest augmenting 
path has length $n - \Theta(d^2)$.

\paragraph{Construction for small, even $d$.}
When $d$ is even we cannot use the 
same balloon gadget, and in fact must connect the balloon gadgets.  
A $t$-fold even-$d$ balloon gadget is constructed as in \Cref{fig:even-small-d}.
One may obtain it by modifying $t$ disjoint odd-$d$ gadgets as follows.
\begin{itemize}
    \item Increase the ``balloon'' to $d+2$ vertices $\{v_1,\ldots,v_{d+2}\}$.
    \item Add $d-2$ unmatched edges joining the $u'$ vertex to $v_1,\ldots,v_{d-2}$,
    and add $(d(d+2)-(d-2))/2$ edges inside $\{v_1,\ldots,v_{d+2}\}$ to make it $d$-regular, with $(d+2)/2$ of them being matched edges.
    \item Keep the matched edge $\{u,u'\}$ in each gadget, and add unmatched edges from the ``$u$'' vertex in each gadget to the ``$u'$'' vertex in the successor gadget.  
\end{itemize}

Observe that all vertices in the $t$-fold balloon gadget have degree $d$, except for the $u$-vertices, which each 
have free degree $d-2$.

\begin{figure}
    \centering
    \includegraphics[scale=0.3]{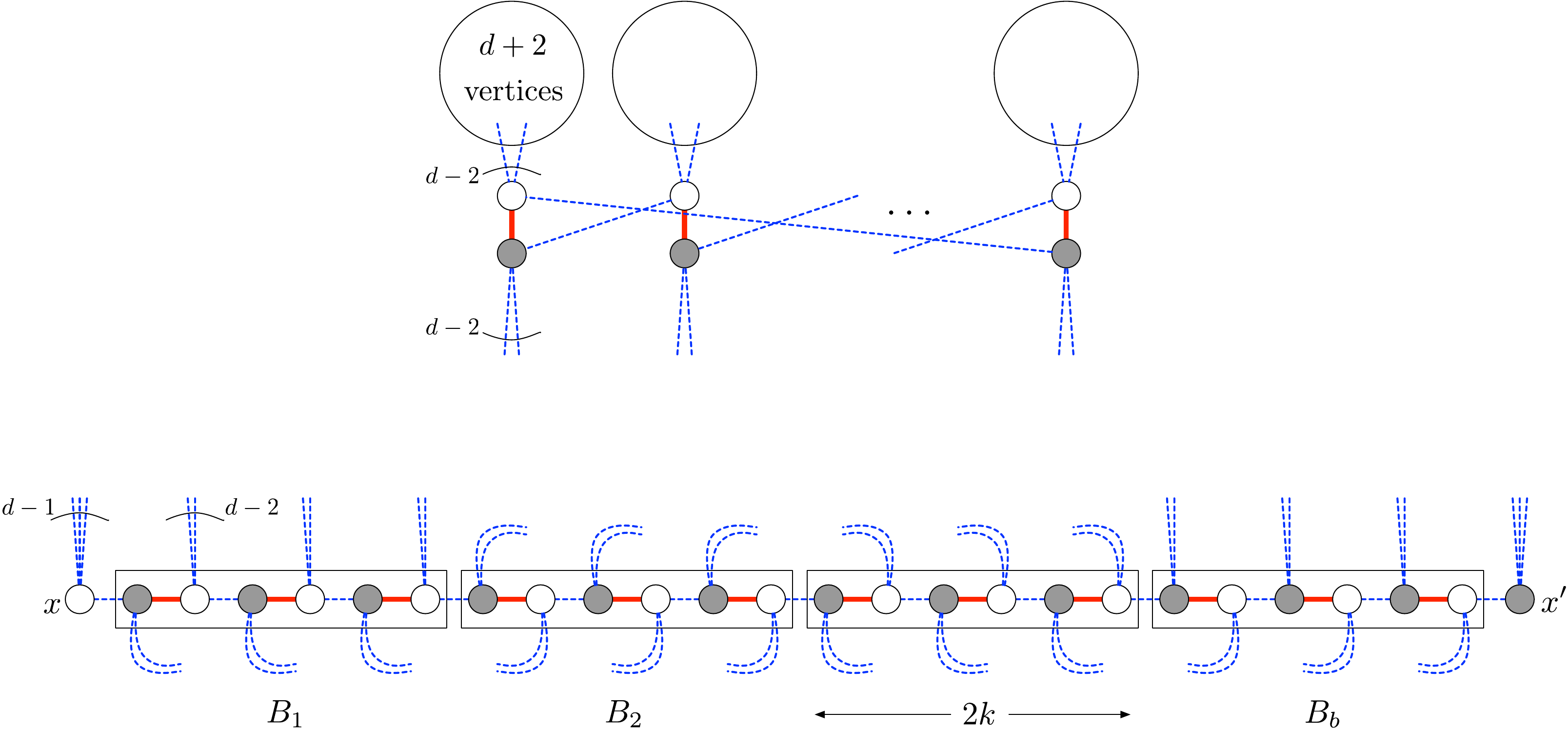}
    \caption{The graph for even $d$ and $d=O(\sqrt{n})$.  Ordinarily we would
    delete $(d-4)/2$ edges joining $B_1$ and $B_2$.  As $d=4$ in this example,
    it is the unique case where no edges need to be deleted.}
    \label{fig:even-small-d}
\end{figure}

The free degree of the $t$-fold balloon gadget is $t(d-2)$ while the free degree
of $V(P)$ is still $2(d-1 + k(d-2))$, 
which is generally not divisible by $d-2$.
To fix this, delete any $(d-4)/2$ 
edges joining $B_1$ and $B_2$, which increases the free degree of $V(P)$ by $d-4$.
Now
\begin{align*}
2(d-1 + k(d-2)) + d-4 
&= (d-2)(2k + 3)\\
&= t(d-2),
\end{align*}
which holds for $k=d-2$ and $t=2k+3$.

The only substantive change to this graph
relevant to augmenting paths are the addition
of edges joining ``$u$'' and ``$u'$'' in successive balloon gadgets.  One may verify that they cannot be used in an augmenting path from $x$ to $x'$.

\medskip
The total number of vertices is
\[
|V(P)| + t(d+4) = |V(P)| + (2d-1)(d+4)
\]
so the unique augmenting path still 
has length $n-\Theta(d^2)$.

\paragraph{Routers.}
In the small-$d$ constructions the 
balloon gadgets contain $\Omega(d^2)$ 
vertices, limiting their applicability to $d=O(\sqrt{n})$.  Another key constraint
is that $k \geq d-2$ so that it is possible to achieve the desired degree of $V(P)$ vertices with edges directly joining $B_i$ and $B_{i+1}$.  
We fix both of these problems with \emph{router gadgets}.  
A $(d,\ell)$-router is obtained from 
$K_{d,d}$ by choosing any perfect matching
of size $d$, then removing 
$\ell\leq d(d-1)$ unmatched edges, leaving
free degree $\ell$ on either side of the bipartition,
and such that the free degree of each individual vertex is $\floor{\ell/d}$
or 
$\ceil{\ell/d}$.
See \Cref{fig:router}.

\begin{figure}
    \centering
    \includegraphics[scale=0.35]{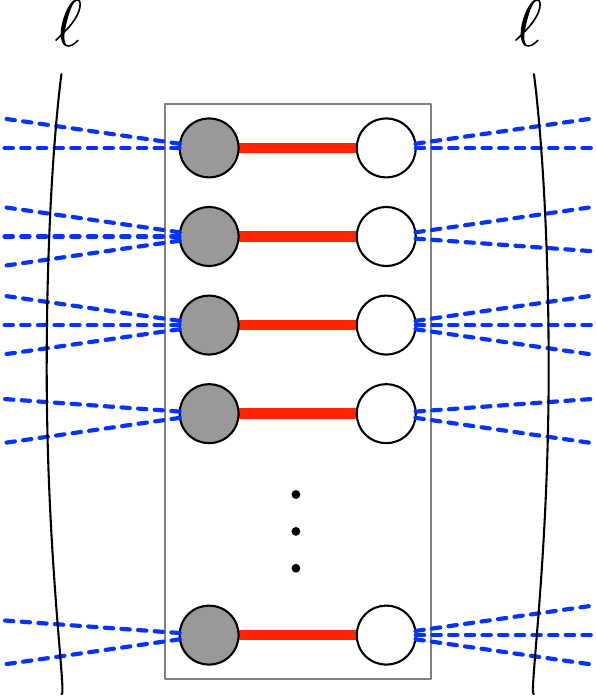}
    \caption{A $(d,\ell)$-router gadget.  
    There are $\ell \leq d(d-1)$ half-edges on either side of the bipartition, whose other endpoint is somewhere outside the gadget.
    Unmatched edges internal to the gadget are not shown.}
    \label{fig:router}
\end{figure}

\paragraph{Construction for large, odd $d$.}
As in the previous constructions,
we connect the odd vertices of $B_i$
to the even vertices of $B_{i+1}$, 
but now through a $(d,\ell)$-router
$R_i$.  Each of $R_2,\ldots,R_{b-1}$
are $(d,k(d-2))$-routers, while
$R_1$ is a $(d,k(d-2)-r)$-router,
for some small $r$ that is introduced
for divisibility issues.
See \Cref{fig:odd-large-d}.

\begin{figure}
    \centering
    \includegraphics[scale=0.3]{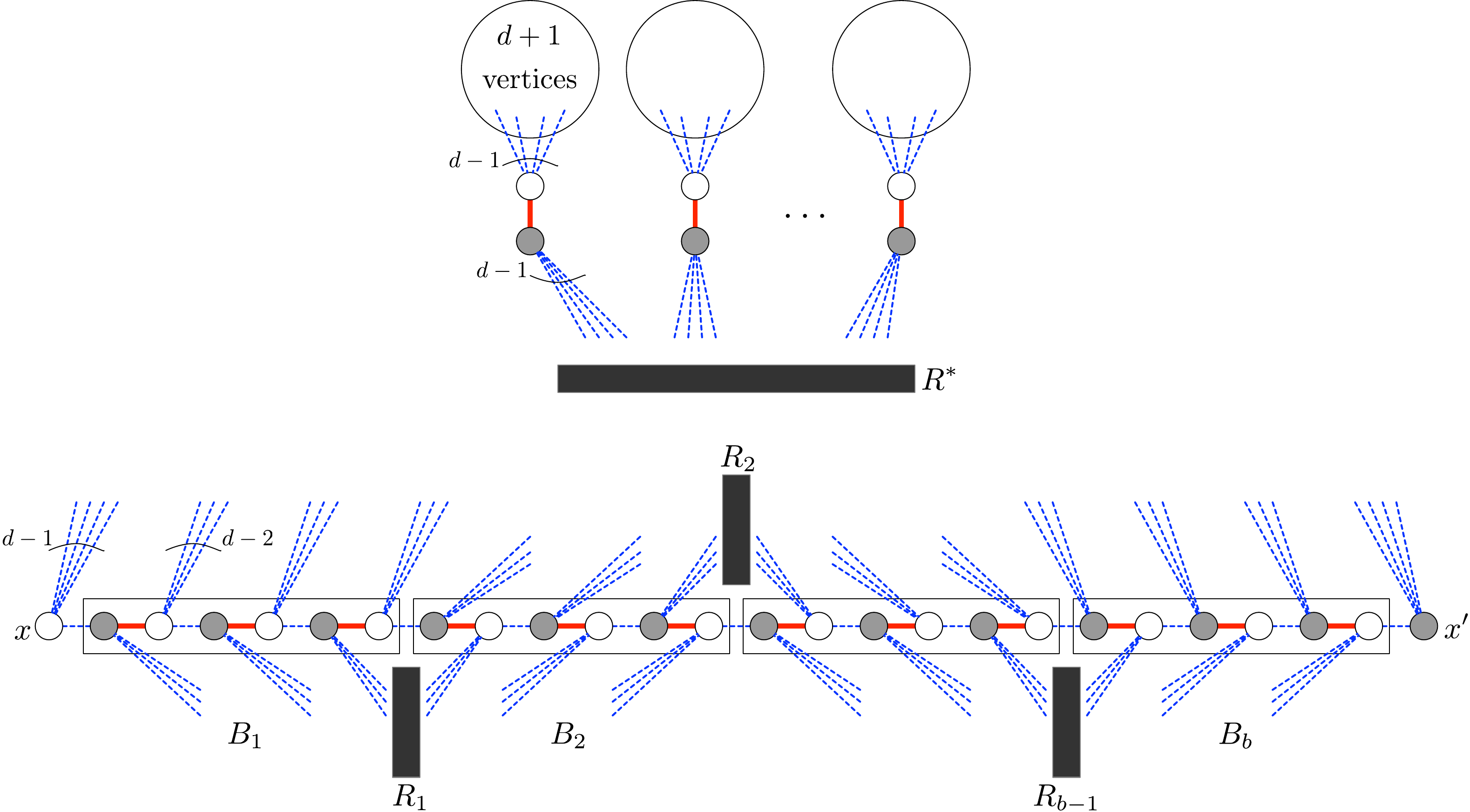}
    \caption{The graph for odd $d=\Omega(\sqrt{n})$.  Black rectangles represent $(d,\ell)$-routers.}
    \label{fig:odd-large-d}
\end{figure}

The free degree of $V(P)$ is therefore
$2(d-1 + k(d-2)) + 2r$, while the free
degree of the balloon gadgets is $t(d-1)$.
We choose $r < (d-1)/2$ such that
\begin{align*}
2(d-1 + k(d-2)) + 2r
&= (d-1)(2 + 2k) + 2(r-k) \\
&= t(d-1),
\end{align*}
We connect $V(P)$ to the balloons 
via a $(d,t(d-1))$-router $R^*$.
Since $t=2k+O(1)$, this is feasible
whenever $k\leq d/2-O(1)$.
Setting $b=k$, 
the total number of vertices is
\begin{align*}
n &=
(2 + 2kb) + 2bd + t(d+3)\\
&= 2(k^2 + 1) + 2kd + (2k+O(1))(d+3).
\end{align*}
Since $k<d/2-O(1)$, $n = \Theta(kd)$
and the length of the augmenting
path $P$ is therefore $2k^2+1 = \Theta((n/d)^2)$.

\paragraph{Construction for large, even $d$.}
The changes are analogous to the small-$d$ case.  
Using a $t$-fold even-degree balloon gadget, 
we choose $r=(d-4)/2$ to equalize 
\[
2(d-1 + k(d-2)) + 2r = t(d-2).
\]
Similar to the previous case, 
$t=2k+3$, $b=k$, 
$n = 2+2kb + 2bd + t(d+4)$,
$k=\Theta(n/d)$
and the rest of the analysis 
is the same.

\section{Nearly Regular Graphs}\label{sect:nearreg}

In this section we extend our results to nearly regular graphs, specifically graphs for which 
\begin{equation}\label{eq:nreg}
    \forall v \in V(G), d \le \deg(v) \le \Delta.
\end{equation}
Let us first consider an extreme case of a graph where the degrees fall in such a range: the complete bipartite graph $K_{d,\Delta}$. Each vertex on the left has degree $\Delta$, each vertex on the right has degree $d$ and when we have matched all the vertices on the left to get a maximum matching, $\Delta -d$ unmatched vertices remain on the right, with no augmenting path. The fraction of unmatched vertices remaining is $\frac{\Delta - d}{\Delta +d}$. Recalling that the limiting value from \Cref{thm:augpath-nonbipartite} was 
$\frac{1}{d+1}$ in the $d$-regular case, 
we define 
\begin{align*}
\tau_{d,\Delta} &= \max\left\{ \frac{1}{d+1}, \frac{\Delta - d}{\Delta +d}, \right\}
\intertext{or equivalently}
\tau_{d,\Delta} &= 
\begin{cases}
    \frac{1}{d+1} & \mbox{when } \Delta -d \le 2,\\
    \istrut{6}\frac{\Delta - d}{\Delta +d}  & \mbox{when } \Delta -d > 2,
\end{cases} 
\end{align*}
where the two formulas agree when $\Delta -d = 2$.
We will show the existence of short augmenting paths in all graphs satisfying \eqref{eq:nreg}, whenever the unmatched fraction exceeds $\tau_{d, \Delta}$. Specifically,

\begin{theorem}\label{thm:nearreg}
    Let $G$ be a graph on $n$ vertices, with degrees in $[d,\Delta]$ and let $M$ a matching in $G$ such that
\begin{equation}
    \label{eq:M-size-near-regular}
    |M| < \frac{n}{2}\big(1 - \tau_{d,\Delta}\big).
\end{equation}
Then there exists  an augmenting path of length $O(n/d)$.  
\end{theorem}

In order to prove the theorem, we do the same layered construction as in Section~\ref{sect:nonbipartite-graphs}, \emph{i.e.,} we construct the sets $\{W_i\}_{i\ge 0}$ and $\{B_i\}_{i\ge 1}$ where   $W_0$ is the set of unmatched vertices, each $B_i$ is the set of vertices with at least two neighbors in $W_0\cup\cdots\cup W_{i-1}$, 
and each $W_i$ is the set of matched neighbors of vertices in $B_i$. 
We will show that this construction has analogous properties to those proved in Lemmas~\ref{lem:nonbip-properties} and~\ref{lem:three-layer-lower-bound}, and thereby obtain the desired result.
We begin with a partial analog of Lemma~\ref{lem:nonbip-properties}, leaving out the non-emptiness of the $B_i$s.

\begin{lemma}
    Suppose all the vertices of $G$ have degrees in $[d,\Delta]$. Let $M$ be a maximal matching in $G$, $W_0$ be the free vertices with respect to $M$, and 
    \[ 
    n_0 = |W_0| > n \tau_{d,\Delta}.
    \]
    If there are no augmenting paths of length less than $4(L+1)$, then the sets 
    $\{W_i\}_{0 \le i \le L }$ and 
    $\{B_i\}_{1\le i \le L}$ constructed as in Section~\ref{sect:nonbipartite-graphs}, satisfy parts (1)--(4) of Lemma~\ref{lem:nonbip-properties}, and for $1\le i \le L$,
    $|W_i| = |B_i|$.
\end{lemma}

\begin{proof}
    Parts (1)--(4) of Lemma~\ref{lem:nonbip-properties} do not use regularity, and therefore go through exactly as in that proof.  
    Moreover part (4) 
    implies that $W_i \cap B_i = \emptyset$, 
    and $|W_i| = |B_i|$.  
\end{proof}

To prove Theorem~\ref{thm:nearreg}, 
we must prove analogues of \Cref{lem:nonbip-properties}(5)
and \Cref{lem:three-layer-lower-bound}, 
specifically that each $B_i$ is non-empty
and
that any three consecutive layers 
contain more than $d$ vertices.
First we fix some notation that will be helpful to prove these claims. 

\begin{itemize}
    \item For $V_1, V_2 \subset V$, let $\edges(V_1, V_2)$ denote the number of edges between $V_1$ and $V_2$. 
\end{itemize}
For a fixed $i$,
\begin{itemize}
    \item As before, let $n_i := |B_i| = |W_i|$.
    \item Let  $B' = B_{<i}$, $W' = W_{<i}$  and  $b= |B'|$. Then, since for all $j\ge 1$, $|W_j| = |B_j|$, it follows that $|W'| = n_0 +b$. 
    \item Let $\gamma(b) = n - (d+1) n_0 + (\Delta-d-2) b $
    \item Let $R= V\setminus (W' \cup B')$. 
    \end{itemize}

\begin{lemma}\label{lem:gamma}
    Let $G$, $M$ be defined as in the statement of Theorem~\ref{thm:nearreg}, and layers $W_0,\ldots,B_{i-1},W_{i-1}$ be constructed as usual.  Then $\gamma(b) < 0$.
\end{lemma}
\begin{proof}
    By our hypothesis, $n_0 > \tau_{d,\Delta} n$. We consider two cases depending on whether 
    $\Delta -d$ is bigger than 2.

\medskip 

\par \noindent 
{\bf Case 1:} $\Delta-d \le 2$. In this case $\tau_{d,\Delta} = \frac{1}{d+1}$, and $\Delta-d -2\le 0$. We have
\[
\gamma(b) = n - (d+1) n_0 + (\Delta-d-2) b  \le n -  (d+1) n_0 = n - \frac{n_0}{\tau_{d,\Delta}} <0.    
\]

\par \noindent 
{\bf Case 2:} $\Delta-d > 2$. In this case $\tau_{d,\Delta} = \frac{\Delta - d}{\Delta +d}$, which implies $(\Delta +d) n_0 > (\Delta-d) n$.
As observed in the paragraph above, $|W'| = n_0+b$, and since $W'$ and $B'$ are disjoint, we have $|W'\cup B'| = n_0 + 2b \le n$ so that $b \le \frac{n-n_0}{2} $. Therefore
\begin{align*}
    2\gamma(b) &= 2n - 2(d+1) n_0 + 2(\Delta-d-2) b  \\
    &\le 2n -  2(d+1) n_0 + (\Delta-d-2)(n-n_0) \\
    &=  (\Delta -d)n - (\Delta+d) n_0 \\
    &<0 . \qedhere
\end{align*}
\end{proof}

\begin{lemma}
    $B_i$ is non-empty, for all $1 \le i \le L$.
\end{lemma}
\begin{proof}
Suppose, for the purpose of obtaining a contradiction, that $B_i$ is empty. Then, by definition of $B_i$, every vertex in $R$ has at most one neighbor in $W'$. Therefore,
\begin{align}
    \edges(W', R) &\le |R| = n - n_0 -2b.\label{eq:edgeswr}
\intertext{Recalling that $W'$ is an independent set, we also know that}
\edges(W', R) &= \sum_{v \in W'} \deg(v) - e(W', B').\nonumber
\intertext{Using the minimum degree in $W'$ and the maximum degree in $B'$, we may also bound $e(W',R)$ as}
\edges(W', R) 
    &\ge d|W'| - \Delta |B'|\nonumber\\
    &=d (n_0 +b) -\Delta b \nonumber\\
    &= d n_0  -(\Delta-d) b\label{eq:edgeswr-lowerbound}
\intertext{Combining \eqref{eq:edgeswr} and \eqref{eq:edgeswr-lowerbound}, we get}
n-n_0 - 2b &\ge dn_0 -(\Delta-d) b.
\end{align}
Bringing all the terms to one side, we see that this implies $\gamma(b) \ge 0$ which contradicts Lemma~\ref{lem:gamma}.
Thus we have established that $n_i = |B_i| >0$.
\end{proof}

\begin{lemma}
    If $i<L$ then $n_{i-1} +n_i +n_{i+1} > d+1$.
\end{lemma}

\begin{proof}
Let $x= \edges(R, B')$ and $y = \edges(R, W')$.
We know that $W'$,  $B'$ and $R$ are pairwise disjoint and that $W'$ is an independent set, 
which implies
\begin{align*}
y &= \sum_{v \in W'} \deg(v) - \edges(W', B')  \ge d|W'| - \edges(W', B') = d(n_0 +b) - \edges(W', B')
\intertext{and}  
x &e= \sum_{v \in B'} \deg(v) - \edges(W', B') - 2\edges(B',B') \le \Delta |B'| -  \edges(W', B') = \Delta b -  \edges(W', B').
\intertext{Combining these and the definition of $\gamma(b)$
we conclude that}
 x &\le y - dn_0  + (\Delta -d)b  = y + n_0 + 2b -n +\gamma(b).
\end{align*}
Noting that $|R| = n - (n_0+2b)$ and $\gamma(b) <0$, we have
\begin{align}
    x &\le  y -|R|.\label{eq:xub}
\intertext{Let $R' = R\setminus (B_i \cup W_i \cup B_{i+1})$ and $r'= \edges(W_i, R')$. Recall 
that there are no edges between $W'$ and $W_i$, 
and that every vertex in $W_i$ has degree at 
least $d$.  We have}
dn_i &\le \sum_{v \in W_i} \deg(v) \nonumber\\
&= \edges(W_i, B_i) + \edges(W_i,B_{i+1} ) + \edges(W_i, R') + \edges(W_i, B') \nonumber\\
&\le  n_i^2 + n_i n_{i+1} + r' + \edges(W_i, B').\nonumber
\intertext{Since $W_i \subset R$ we have}
dn_i &\le  n_i^2 + n_i n_{i+1} + r' + x. \label{eq:dniub}
\intertext{At this point we follow the proof of Eq.~\eqref{eq:yub}, which did not use regularity but only the definitions of $B_i$ and $B_{i+1}$,
to conclude that}
    y &\le n_{i-1}n_i + n_i + n_{i+1} + |R'|-r'.\label{eq:yub2}
\end{align}
Combining equations \eqref{eq:xub}, \eqref{eq:dniub} and \eqref{eq:yub2}, we have
\begin{align*}
    dn_i &\le n_i^2 + n_i n_{i+1} + r' + x  & \mbox{by \eqref{eq:dniub}}\\
&< n_i^2 + n_i n_{i+1} + r' + y -|R|  & \mbox{by \eqref{eq:xub}}\\
&\le n_i^2 + n_i n_{i+1} + r' + n_{i-1}n_i + n_i + n_{i+1} + |R'|-r' -|R|  & \mbox{by \eqref{eq:yub2}}\\
&= n_i(n_i + n_{i+1} +n_{i-1}) + (n_i + n_{i+1} +|R'|)  -|R| & \mbox{rearranging terms}\\
&= n_i(n_i + n_{i+1} +n_{i-1}) -n_i  & \mbox{Since $|R'| = |R| -2n_i -n_{i+1} $}\\
&= n_i(n_{i-1} + n_i + n_{i+1} -1) 
\end{align*}
Dividing by $n_i >0$ and rearranging gives the desired result.
\end{proof}

The proof of Theorem~\ref{thm:nearreg} now follows exactly the same lines as the proof of Theorem~\ref{thm:augpath-nonbipartite}.

\begin{corollary}
    Let $G$ be a graph with degrees in $[d,\Delta]$ and let $r = \Delta -d +1$. Any blocking flow-type matching algorithm finds a maximum matching in $G$ in time $O\big(r n^2\big)$
\end{corollary}
\begin{proof}
    As in the proof of Corollary~\ref{cor:nonbipartite-max-card-matching}, we note that after $k$ phases of a blocking flow algorithm, the shortest augmenting path is of length at least $k$. It follows by Theorem~\ref{thm:nearreg} that after $O(n/d)$ phases, at most $\tau_{d,\Delta} n $ unmatched vertices remain, so that there can be at most $\tau_{d,\Delta} n/2$ more phases. Thus the algorithm has $O\left(\frac{n}{d} + \tau_{d,\Delta} n\right)$ phases. Each phase can be run in linear time in the number of edges, $m \le \Delta n$. Thus the running time is 
    \[
    O\left(\frac{n}{d} + \tau_{d,\Delta} n\right)\cdot \Delta n = O\left(\left(\frac{\Delta}{d} + \Delta\tau_{d,\Delta} \right) n^2\right)
    \]
    The result then follows from observing that $\Delta/d$ and $\Delta\tau_{d,\Delta}$ are both at most $2r$.
\end{proof}

\end{document}